\documentclass[11pt,a4paper]{article}

\usepackage{microtype}

\usepackage{amssymb,amsmath,amsthm}
\usepackage{graphicx}
\usepackage{color}
\usepackage{paralist}
\usepackage{titlesec}

\title{
	Integrality Gap of the Configuration LP for Restricted Max-Min Fair Allocation\thanks{Research supported by the Research Grants Council, Hong Kong, China (project no.~16200317).}
}

\author{Siu-Wing Cheng  \and Yuchen Mao}
\date{Department of Computer Science and Engineering, HKUST, Hong Kong,\\
	\texttt{\{scheng, ymaoad\}@cse.ust.hk}}

\newcommand{\eps} {\varepsilon}
\newenvironment{emromani} {%
\begin{enumerate}}
{\end{enumerate}}

\newtheorem{claim}{Claim}
\newtheorem{theorem}{Theorem}
\newtheorem{lemma}{Lemma}

\begin{document}

\maketitle

\begin{abstract}
	The max-min fair allocation problem seeks an allocation of resources to players that maximizes the minimum total value obtained by any player.  Each player $p$ has a non-negative value $v_{pr}$ on resource $r$.  In the restricted case, we have $v_{pr}\in \{v_r, 0\}$. That is, a resource $r$ is worth value $v_r$ for the players who desire it and value 0 for the other players.   In this paper, we consider the configuration LP, a linear programming relaxation for the restricted problem.  The integrality gap of the configuration LP is at least $2$.  Asadpour, Feige, and Saberi proved an upper bound of $4$.  We improve the upper bound to $23/6$ using the dual of the configuration LP.  Since the configuration LP can be solved to any desired accuracy $\delta$ in polynomial time, our result leads to a polynomial-time algorithm which estimates the optimal value within a factor of $23/6+\delta$.
\end{abstract}

\section{Introduction}

\paragraph*{Background.} The \emph{max-min fair allocation} problem is to distribute resources to players in a fair manner.  There is a set $P$ of $m$ players and a set $R$ of $n$ indivisible resources.  Resource $r \in R$ is worth a non-negative value $v_{pr}$ for player $p$.  An allocation is a partition of $R$ into disjoint subsets $\{\textrm{$C_p$ : $p \in P$}\}$, i.e., player $p$ is given the resources in $C_p$.  The value of an allocation is $\min_{i \in P} \sum_{r \in C_p} v_{pr}$.  The goal is to find an allocation with maximum value.  

Bez\'{a}kov\'{a} and Dani~\cite{BD05} proved that no polynomial-time algorithm can offer an approximation ratio smaller than 2 unless $\mathcal{P} = \mathcal{NP}$. They also showed that the assignment LP has an unbounded integrality gap.  Bansal and Sviridenko~\cite{BS06} proposed a stronger LP relaxation, the configuration LP, and showed that the integrality gap is at least $\Omega(\sqrt{m})$.  They also showed that the configuration LP can be solved to any accuracy $\delta >0$ in polynomial time.  Asadpour and Saberi~\cite{AS07} later proved that the integrality gap is at most $O(\sqrt{m}\log^3{m})$ by proposing a polynomial-time rounding scheme. Saha and Srinivasan~\cite{SS10} improved the bound to $O(\sqrt{m\log m})$. 

We focus on the restricted case in this paper.  In the restricted case, every resource $r \in R$ has a value $v_r$ such that $v_{pr} \in \{v_r,0\}$.  That is, if player $p$ desires resource $r$, then $v_{pr} = v_r$; otherwise, $v_{pr} = 0$.  The approximation ratio is still at least 2 in the restricted case unless $\mathcal{P} = \mathcal{NP}$~\cite{BD05}.  The configuration LP turns out to be a useful tool for the restricted case.  A simple instance can show that the integrality gap of the configuration LP is at least $2$~\cite{JR17}.  
Bansal and Sviridenko~\cite{BS06} obtained an $O\bigl(\frac{\log\log m}{\log\log\log m}\bigr)$-approximation algorithm by rounding the configuration LP.  Feige~\cite{F08} proved that the integrality gap of the configuration LP is bounded by a constant although the constant factor is large and unspecified.  His proof was later made constructive by Haeupler, Saha, and Srinivasan~\cite{HSS11}, and results in a constant-approximation algorithm.  Inspired by Haxell's technique for hypergraph matching~\cite{H95}, Asadpour, Feige and Saberi~\cite{AFS12} proved that the integrality gap of the configuration LP is at most 4, but their construction of the corresponding approximate allocation is not known to be in polynomial time.  By extending the idea of Asadpour et al., Annamalai, Kalaitzis and Svensson~\cite{AKS17} developed a polynomial algorithm to construct an allocation with an approximation ratio of $6 + 2\sqrt{10} + \delta$ for any constant $\delta > 0$.  Their analysis involves the configuration LP. Recently, we improved the approximation ratio to $6 + \delta$ for any constant $\delta > 0$ using a pure combinatorial analysis~\cite{CM18}.

Chan, Tang, and Wu~\cite{CTW16} considered an even more constrained \emph{$(1,\eps)$-restricted case} in which each resource $r$ has value $v_r \in \{1,\eps\}$ for some fixed constant $\eps \in (0,1)$.  They showed that the integrality gap of the configuration LP in this case is at most 3.

The configuration LP is also helpful for the restricted case of the classic problem {\sc scheduling on unrelated machines to minimize makespan}, in which people want to minimize the maximum load of any machine.  In a seminal paper, Lenstra, Shmoys, and Tardos~\cite{LST90} proposed a 2-approximation algorithm for this scheduling problem by rounding the assignment LP.  They also showed that no polynomial time algorithm can do better than $3/2$ unless $\mathcal{P} = \mathcal{NP}$.  An important open question is whether this can be improved.  Inspired by the techniques used in~\cite{AFS12}, Svensson~\cite{S12} showed that the integrality gap of configuration LP for the restricted case of the scheduling problem mentioned above is at most $33/17$.  Jansen and Rohwedder~\cite{JR17} improved the bound to $11/6$.

\paragraph*{Our contributions and techniques.} In~\cite{AFS12}, Asadpour, Feige and Saberi construct the approximate allocation by building a hypergraph matching. The value of hyperedges (that is, the total value of resources covered by an edge) plays an important role in the analysis. They use the maximum possible value to bound the value of every hyper-edge that is chosen by their construction algorithm.  This is pessimistic because not all edges can reach this maximum value.  Our innovation lies in using the dual of the configuration LP to amortize the values of addable edges.  When constructing the dual solution, we scale down the value of large resources.  This reduces the (dual) value of the edges that contain lots of large resources.  Our main result is the following.

\begin{theorem}
	For the restricted max-min fair allocation problem, the integrality gap of the configuration LP is at most $23/6$.
\label{thm:main}
\end{theorem}

\paragraph*{Configuration LP and its dual.}
For player $p$, we define $R_p = \{\text{$r \in R$ : $v_{pr} > 0$}\}$, that is, $R_p$ is the set of resources player $p$ is interested in.  For a set $S$ of resources, we define its total value $value(S) = \sum_{r\in S}v_r$. 

Let $T$ be our target value. A \emph{configuration} for a player $p$ is any subset $S \subseteq R_p$ such that $value(S) \geq T$.  The set of all configurations for player $p$ is defined as 
\begin{equation*}
	{\cal C}_p(T) = \{\text{$S\subseteq R_p$ : $\mathit{value}(S) \geq T$}\}.
\end{equation*}
Given a target $T$, the configuration LP, denoted as $\mathit{CLP}(T)$, associate a variable $x_{p,C}$ for each player $p$ and each configuration $C$ in ${\cal C}_p(T)$. Its constraints ensure that each player receives at least $1$ unit of configurations, 
and that for any $r\in R$, the total unit of $r$ that are assigned to players is at most 1.
The optimal value of the configuration LP is the largest $T$ for which $CLP(T)$ is feasible. We denote this optimal value by $T^*$.  Although the configuration LP can have an exponential number of variables, it can be approximately solved to an arbitrary accuracy $\delta$ in polynomial time~\cite{BS06}.

The configuration LP is given on the left of Figure~\ref{fig:lp}. Although there is no objective function to optimize, one can view it as a minimization of a constant objective function. On the right of the figure is the dual LP~\cite{BS06}.  If the configuration LP is feasible, then the optimal value of its dual must be bounded from above.  

\begin{figure}
	\begin{center}
		\begin{minipage}[t][11em][t]{0.4\textwidth}
		\centering\textbf{Primal}
		\begin{alignat*}{3}
    	&\quad 	& \sum_{C\in {{\cal C}_p(T)}}x_{p,C} &\geq 1
				& \quad &\forall p\in P\\
		&		&\sum_{p\in P}\sum_{C\in {\cal C}_p(T): r\in C} x_{p, C} &\leq 1
				& \quad &\forall r\in R\\
		&		& x_{p, C}&\geq 0
		\end{alignat*}
		\end{minipage}
		\hfill
		\begin{minipage}[t][11em][t]{0.5\textwidth}
		\centering\textbf{Dual}
		\begin{alignat*}{3}
		\max&\quad & &\sum_{p\in P} y_p - \sum_{r \in R} z_r\\
		s.t.&\quad & y_p&\leq \sum_{r\in C}z_r
			& \quad &\forall p\in P, \forall C\in {\cal C}_p(T)\\
			&		& y_p&\geq 0 & \quad &\forall p\in P\\
			&		& z_r&\geq 0 & \quad &\forall r\in R
		\end{alignat*}
		\end{minipage}
	\end{center}
	\caption{The configuration LP and its dual.}
	\label{fig:lp}
\end{figure}

\section{Hypergraph Matchings}
Without loss of generality, we assume that $T^* = 1$. To prove our upper bound of $23/6$ on the integrality gap, we aim to find an allocation in which every player receives at least $\lambda$ value of resources, where $\lambda = 6/23$.

We call a resource $r$ \textbf{fat} if $v_r \geq \lambda$, and \textbf{thin}, otherwise.  Recall that $R_p$ is the set of resources in which player $p$ is interested.   We use $R^f_p$ and $R^t_p$ to denote the sets of fat and thin resources, respectively, in which player $p$ is interested.  Note that a player $p$ will be satisfied if it receives a single fat resource from $R^f_p$ or a bunch of thin resources of total value at least $\lambda$ from $R^t_p$.

Consider the following bipartite hypergraph $G = (P \cup R, E)$.  There are two kinds of edges in $G$: \emph{fat edges} and \emph{thin edges}.  For each $p\in P$ and each $r\in R^f_p$, $(p,\{r\})$ is a fat edge of $G$.  For each $p\in P$ and each \emph{minimal} subset $S$ of $R^t_p$ such that $\mathit{value}(S)\geq \lambda$, $(p, S)$ is a thin edge of $G$.  By minimality, we mean that deleting any resource from $S$ will result in $\mathit{value}(S) < \lambda$.  For any edge $e = (p, S)$ of $G$, we say $p$ and resources in $S$ are covered by $e$. A player (resources) is covered by a subset of edges in $G$ if it is covered by some edge in this subset.

A matching of $G$ is a set of edges that do not share any player or any resource.  A perfect matching of $G$ is a matching that covers all players.  A perfect matching of $G$ naturally induces an allocation in which every player receives at least $\lambda$ worth of resources.  To prove Theorem~\ref{thm:main}, all we need to do is to show that $G$ has a perfect matching.

\begin{theorem}
	The bipartite hypergraph $G$ has a perfect matching.
	\label{thm:matching}
\end{theorem}

To prove Theorem~\ref{thm:matching}, it suffices to show that given any matching $M$ of $G$,
if $|M| < m$, we can always increase the matching size by $1$.  In this section, we provide a local search algorithm to achieve this.  This algorithm works in the same spirit as the one in~\cite{AFS12}.

Let $M$ be a matching of $G$ such that $|M| < m$.  Let $p_0$ be a player that is not yet matched by $M$.  The algorithm will return a matching $M'$ such that $M'$ matches $p_0$ and all players matched by $M$.

To match $p_0$, we pick an edge $(p_0, S_0)$ from $G$.  If no resource in $S_0$ is used by $M$, we simply match $p_0$ using this edge.  Assume that some edge in $M$, say $(q, S)$, uses some resource in $S_0$.  We say that $(q, S)$ \textbf{blocks} $(p_0, S_0)$ in this situation.  In order to free up the resources in $S$ and make $(p_0, S_0)$ unblocked, we need to find a new edge to match $q$.  Now $q$ has a similar role as $p_0$. 

For the convenience of analysis, we use terminologies that are similar to those in~\cite{JR17} rather than~\cite{AFS12}.  Our algorithm maintains a sequence $(B_1, \ldots, B_{\ell})$ of blockers in chronological order of their addition.  That is, $B_i$ is added before $B_j$ for $i < j$.  A \textbf{blocker} $B_k$ is a tuple $(x_k, Y_k)$ where $x_k$ is an edge we hope to add to the matching, and $Y_k$ is the subset of edges in $M$ that prevent us from doing so, i.e., share resources with $x_k$. We call edges in $Y_k$ \textbf{blocking edges}.  We say a blocker $B_k = (x_k, Y_k)$ is \textbf{removable} if $Y_k = \emptyset$.  A player $p$ is \textbf{active} if $p$ is covered by some blocking edge (from $Y_k$) in some blocker $B_k$, and we say that blocker $B_k$ \textbf{activates} player $p$.  We will prove in Lemma 1 that the blocking edges in the blockers do not share any players, so an active player is activated by exactly one blocker. The player $p_0$ is always active until it is matched (and the algorithm terminates).  For a sequence $(B_1, \ldots, B_{\ell})$ of blockers, we use $\mathit{resource}(B_{\leq \ell})$ to denote the set of the resources covered by $\{x_1, \ldots, x_\ell\} \cup Y_{\leq \ell}$ where $Y_{\leq \ell} = \bigcup_{i = 1}^{\ell}Y_i$.

The sequence of blockers is built inductively. Initially, only $p_0$ is active, and the sequence is empty.  Let $\ell$ be the number of blockers in the sequence.  Consider the construction of the $(\ell+1)$-th blocker.  We call an edge $(q, S)$ \textbf{addable} if player $q$ is active and $S\cap \mathit{resource}(B_{\leq \ell}) = \emptyset$.  We arbitrarily pick an addable edge as $x_{\ell+1}$.  Then we add to $Y_{\ell+1}$ all the edges from $M$ that share some resource(s) with $x_{k+1}$. The pair $(x_{\ell+1}, Y_{\ell+1})$ form a new blocker $B_{\ell+1}$.  

After constructing a new blocker, the algorithm checks whether some blocker is removable.  If so, the algorithm will remove the removable blockers as well as some other blockers, and update $M$.  When there is no removable blocker, the algorithm resumes the construction of the next new blocker.

The algorithm alternates between the construction and the removal of blockers until $p_0$ is matched.  The construction and the removal are carried out by the routines {\sc Build} and {\sc Contract} specified in the following.

\vspace{12pt}

\begin{quote}
	{\sc Build}$\left(M, \left(B_1, \ldots, B_{\ell}\right)\right)$
	\begin{enumerate}
		\item Arbitrarily pick an addable edge $x_{\ell+1}$.
		\item Let $Y_{\ell + 1} = \{\text{$e \in M$ : $e$ shares some resources with $x_{l+1}$}\}$.
		\item Append a new blocker $B_{\ell + 1} = (x_{\ell+1}, Y_{\ell + 1})$ to the sequence. $\ell := \ell + 1$.
	\end{enumerate}

	\vspace{8pt}

	{\sc Contract}$(M, (B_1, \ldots, B_{\ell}))$
	\begin{enumerate}
		\item Let $B_k = (x_k, \emptyset)$ be the removable blocker with the smallest index. 
		Let $q$ be the player covered by $x_k$. If $q = p_0$, we update $M := M\cup \{x_k\}$, and the algorithm terminates.  Assume $q \neq p_0$.  Let $B_j = (x_j, Y_j)$ be the blocker that activates $q$.  Note that $j < k$.  Let $e = (q, S)$ be the edge in $Y_j$ that covers $q$.

		\item Update $M := (M\setminus\{e\})\cup \{x_k\}$.  Before the update, $q$ is matched by the edge $e$, and after the update, $q$ is matched by the edge $x_k$. The other players are not affected.

		\item Since $e$ is removed from the matching, it no longer blocks $x_j$. Update $Y_j := Y_j \setminus \{e\}$.

		\item Delete all the blockers that succeed $B_j$, that is, $B_i$ with $i > j$. Set $\ell := j$. 
\end{enumerate}
\end{quote}

\section{Analysis}
We list a few invariants below that are maintained by the algorithm.
\begin{lemma}
Let $M$ be current matching. Let $(B_1, \ldots, B_{\ell})$ be the current sequence of blockers where $B_i = (x_i, Y_i)$. Then the followings hold.
\begin{emromani}[(i)]
	\item For $i_1\neq i_2$, $x_{i_1}$ and $x_{i_2}$ do not share any resource.
	\item For $i\in [1,\ell]$, $x_i$ is blocked by every edge in $Y_i$, but not by any edge in $M\setminus Y_i$.
	\item $Y_1, \ldots, Y_\ell$ are mutually disjoint subsets of $M$.
\end{emromani}
\label{lem:invar}
\end{lemma}
\begin{proof}
	We prove these invariants by induction.  At the beginning of the algorithm, all the invariants trivially hold.  Suppose that all the invariants hold before invoking {\sc Build} or {\sc Contract}, we show that they hold afterwards.

	\vspace{1em}

	{\sc Build}: Let $M$ and $(B_1, \ldots, B_\ell)$ be the matching and the sequence of blockers before calling {\sc Build}.  Let $B_{\ell+1} = (x_{\ell+1}, Y_{\ell+1})$ be the new blocker added to the sequence during {\sc Build}.  By the definition of addable edges, $x_{\ell+1}$ does not share any resource with any of $x_1, \ldots, x_{\ell}$, so invariant (i) is preserved.  Every edge in $M$ that blocks $x_{\ell + 1}$ is added to $Y_{\ell+1}$, so $x_{\ell+1}$ is not blocked by any edge in $M\setminus Y_{\ell+1}$, so invariant (ii) is preserved.  Consider invariant (iii).  Clearly $Y_1, \ldots, Y_{\ell+1}$ are subsets of $M$.  $Y_{\ell+1}$ is the only one that may break disjointness.  By the definition of addable edges, $x_{\ell+1}$ is not blocked by any edge in $Y_1, \ldots, Y_{\ell}$.  By invariant (ii), $x_{\ell + 1}$ is blocked by every edge in $Y_{\ell+1}$.  Therefore, $Y_{\ell+1}$ must be disjoint from $Y_1, \ldots, Y_\ell$. Invariant (iii) is preserved.

	\vspace{1em}

	{\sc Contract}: Let $M$ and $(B_1, \ldots, B_\ell)$ be the matching and the sequence of blockers before calling {\sc Contract}.  Let $B_k = (x_k, Y_k)$ be the removable blocker with the smallest index.  Let $B_j = (x_j, Y_j)$ be the blocker that activates the player covered by $x_k$. 
	Let $e$ be the edge in $Y_j$ that is removed from $M$ when adding $x_k$ to $M$.  Let $Y'_j = Y_j \setminus \{e\}$, $B'_j = (x_j,Y'_j)$, and $M' = (M \setminus \{e\}) \cup \{x_k\}$.  We shall prove that all the invariants hold for $M'$ and $(B_1, \ldots, B_{j-1}, B'_j)$. 
	Since $x_1, \ldots, x_j$ do not change, invariant (i) still holds. Consider invariant (ii).  For $i \in [1, j-1]$, $x_i$ and $Y_i$ do not change, so $x_i$ is still blocked by every edge in $Y_i$.  For $x_j$ and $Y'_j$, since $Y'_j$ is obtained from $Y_j$ by removing the edge $e$ that does not block $x_j$, so $x_j$ is blocked by every edge in $Y'_j$.  Comparing to $M$, $M'$ gains a new edge $x_k$. By inductive assumption, $x_k$ does not share any resource with $x_1,\ldots,x_j$, so $x_k$ does not block any of them.  So invariant (ii) is preserved.  $Y_1, \ldots, Y_{j-1}$ do not change, and $Y'_j$ is a subset of $Y_j$. So by inductive assumption, $Y_1, \ldots, Y_{j-1}, Y'_j$ are mutually disjoint subsets of $M'$. Therefore, invariant (iii) is preserved. 
\end{proof}

To prove that the algorithm eventually extends $M$ to match $p_0$, we need to show two things.  First, we show that, at any time before $p_0$ is matched, there is either a removable blocker or an addable edge.  In the former case, we can call {\sc Contract} to remove some blockers, and in the latter case, we can call {\sc Build} to add a new blocker.  Second, we show that the algorithm terminates after a finite number of steps.

\begin{lemma}
	At any time, there is either a removable blocker or an addable edge.
\label{lem:never-stuck}
\end{lemma}
\begin{proof}
	Let $M$ and $(B_1, \ldots, B_\ell)$ be the current matching and the sequence of blockers, respectively.  Suppose, for the sake of contradiction, that there is neither a removable blocker nor an addable edge.  We show that the dual of the configuration LP (with respect to $T^* = 1$) is unbounded, which implies that the primal $\mathit{CLP}(T^*)$ is infeasible, a contradiction to the definition of $T^*$.

	Consider the following solution of the dual. For each player $p$ and each resource $r$,
	define the variables 

	\begin{align}
		y^*_p &= \left\{
				\begin{array}{ll}
					1 - \frac{4}{3}\lambda &\text{if $p$ is active,}\\
					0 &\text{otherwise}
				\end{array}
		\right. 
		\label{eq:dual_y}\\
		z^*_r &= \left\{
				\begin{array}{ll}
					0 
						&\text{if $r \notin \mathit{resource}(B_{\leq \ell})$ }\\
					1 - \frac{4}{3}\lambda 
						&\text{if $r\in \mathit{resource}(B_{\leq \ell})$ and $r$ is fat}\\
					\min\{v_r, \frac{5}{6}\lambda\} 
						&\text{if $r\in \mathit{resource}(B_{\leq \ell})$ and $r$ is thin}\\
				\end{array}
				\right.
		\label{eq:dual_z}
	\end{align}

	\vspace{4pt}
	\begin{quote}
	\begin{claim}
		The above solution is feasible, i.e., $y^*_p \leq \sum_{r\in C}z^*_r$ for every player $p\in P$ and every configuration $C\in {\cal C}_p(1)$.
	\label{clm:feasible}
	\end{claim}
	\begin{claim}
		The objective function value of the dual is positive, i.e., $\sum_{p\in P} y^*_p - \sum_{r \in R} z^*_r > 0 $.
	\label{clm:positive-obj}
	\end{claim}
	\end{quote}
	
	Let $P = \{p_1, \ldots, p_m\}$ and $R = \{r_1, \ldots, r_n\}$. Since $(y^*_{p_1},\ldots, y^*_{p_m}, z^*_{r_1}, \ldots, z^*_{r_n})$ is a feasible solution, so is $(\alpha y^*_{p_1},\ldots, \alpha y^*_{p_m}, \alpha z^*_{r_1}, \ldots, \alpha z^*_{r_n})$ for any $\alpha > 0$.  As $\alpha \to +\infty$, the objective value goes to $+\infty$ by Claim~\ref{clm:positive-obj}. Hence, the dual LP is unbounded.
\end{proof}

\begin{proof}[\textbf{Proof of Claim \ref{clm:feasible}}]
	If $p$ is not active, then $y^*_p = 0$. Since $z^*_r$ is non-negative, we have $\sum_{r\in C}z^*_r\geq 0 = y^*_p$ for every $C\in {\cal C}_p(1)$.  Assume $p$ is active. So $y^*_p = 1 - 4\lambda/3$.  Let $C$ be any configuration from ${\cal C}_p(1)$.  We show that $\sum_{r\in C}z^*_r\geq 1 - 4\lambda/3$.

	\begin{enumerate}[{Case} 1.]
		\item 
		$C$ contains at least one fat resource, say $r_1$.  Then $r_1$ must appear in $\mathit{resource}(B_{\leq \ell})$; otherwise, $(p, \{r_1\})$ would be an addable edge, contradicting the assumption that there is no addable edge.  Thus $z^*_{r_1} = 1 - 4\lambda/3$ by \eqref{eq:dual_z}, so $\sum_{r\in C}z^*_r \geq z^*_{r_1} = 1 - 4\lambda/3$.

		\item 
		$C$ contains only thin resources, and $C\cap \mathit{resrouce}(B_{\leq \ell})$ contains at least three thin resources $r_1, r_2, r_3$ of values at least $5\lambda/6$. Given that $\lambda = 6/23$, we obtain 
		\begin{equation*}
			\sum_{r\in C}z^*_r \geq z^*_{r_1} + z^*_{r_2} + z^*_{r_3}
			\geq \frac{5\lambda}{6} * 3 = 1 - 4\lambda/3.
		\end{equation*}

		\item 
		$C$ contains only thin resources, and $C\cap \mathit{resource}(B_{\leq \ell})$ contains at most two thin resources of values at least $5\lambda/6$. Since there is no addable edge, every configuration in ${\cal C}_p(1)$ has less than $\lambda$ worth of thin resources that do not belong to $\mathit{resource}(B_{\leq \ell})$. Therefore, 
		\begin{equation*}
			\sum_{r\in C\setminus \mathit{resource}(B_{\leq \ell})} v_r < \lambda.
		\end{equation*}
		Note that $value(C) \geq T^* = 1$ by the definition of ${\cal C}_p(1)$. Therefore, 
		\begin{equation*}
			\sum_{r\in C\cap \mathit{resource}(B_{\leq \ell})} v_r
			= value(C) - \sum_{r\in C\setminus \mathit{resource}(B_{\leq \ell})} v_r
			> 1 - \lambda.
		\end{equation*}
		Recall that $C$ contains only thin resources and $C\cap \mathit{resource}(B_{\leq \ell})$ contains at most two thin resources of values at least $5\lambda/6$. According to~\eqref{eq:dual_z}, for every $r \in C\cap resource(B_{\leq \ell})$, we set $z^*_r:= \min\{v_r, 5\lambda/6\}$ . At most two resources in $C\cap \mathit{resource}(B_{\leq \ell})$ have their $z^*_r = 5\lambda/6 < v_r$, and any other resource in $C\cap \mathit{resource}(B_{\leq \ell})$ has $z^*_r = v_r$.  Hence, for $C\cap \mathit{resource}(B_{\leq \ell})$, the total difference between $v_r$'s and $z^*_r$'s is at most $\lambda/6*2 = \lambda/3$.  Therefore, 
		\begin{equation*}
			\sum_{r\in C}z^*_r \geq \sum_{r\in C\cap \mathit{resource}(B_{\leq \ell})} z^*_r
			\geq \sum_{r\in C\cap \mathit{resource}(B_{\leq \ell})} v_r - \lambda/3 > 1 - \lambda - \lambda/3 = 1 - 4\lambda/3.
		\end{equation*}
	\end{enumerate}
\end{proof}

\begin{proof}[\textbf{Proof of Claim \ref{clm:positive-obj}}]
	For $i\in [1,\ell]$, let $P_i$ be the set of players activated by blocker $B_i = (x_i, Y_i)$.  Note that $P_i$ is exactly the set of players covered by edges in $Y_i$. For $i\in [1,\ell]$, let $R_i$ be the set of resources covered by $\{x_i\}\cup Y_i$. By Lemma~\ref{lem:invar}, all $P_i$'s are mutually disjoint, so are all $R_i$'s. Since $y^*_p = 0$ if $p$ is inactive, we have
	\begin{equation*}
		\sum_{p\in P}y^*_p = y^*_{p_0} + \sum_{i\in [1,\ell]}\sum_{p\in P_i} y^*_p.
	\end{equation*}
	By~\eqref{eq:dual_z}, $z^*_r = 0$ if $r\notin \mathit{resource}(B_{\leq \ell})$. Therefore, 
	\begin{equation*}
		\sum_{r \in R} z^*_r = \sum_{i\in [1,\ell]}\sum_{r\in R_i} z^*_r.
	\end{equation*}
	Combining the above two equations, we obtain
	\begin{equation*}
		\sum_{p\in P} y^*_p - \sum_{r \in R} z^*_r = y^*_{p_0} 
			+ \sum_{i\in [1,\ell]}\left(\sum_{p\in P_i} y^*_p - \sum_{r\in R_i} z^*_r\right).
	\end{equation*}
	Player $p_0$ is always active, so $y^*_{p_0} = 1 - 4\lambda/3 > 0$ by~\eqref{eq:dual_y}. To prove that $\sum_{p\in P} y^*_p - \sum_{r \in R} z^*_r > 0$, it suffices to show that for each blocker $B_i = (x_i, Y_i)$, 
		$\sum_{p\in P_i} y^*_p \geq \sum_{r\in R_i} z^*_r$. 
	Note that $Y_i \neq \emptyset$ for all $i\in [1,\ell]$ because there is no removable blocker.

	\begin{enumerate}[{Case} 1.]
		\item 
		$x_i$ is a fat edge.  Then $Y_i$ must contain exactly one fat edge $(q, \{r_i\})$,
		and $r_i$ is exactly the fat resource covered by $x_i$.  It follows that $P_i = \{q\}$ and $R_i = \{r_i\}$. We have that
		\begin{equation*}
		\sum_{r\in R_i} z^*_r = z^*_{r_i} \overset{\eqref{eq:dual_z}}{=} (1 - 4\lambda/3) 
								 \overset{\eqref{eq:dual_y}}{=} y^*_{q} = \sum_{p\in P_i} y^*_p.
		\end{equation*}

		\item 
		$x_i$ is a thin edge and $|Y_i| \geq 2$.  By the minimality of $x_k$, the total value of resources covered by $x_i$ is less than $2\lambda$.  The blocking edges in $Y_i$ must be thin. By the minimality, for each blocking edge in $Y_i$,
		the total value of its resources that are not covered by $x_k$ is less than $\lambda$. 
		Hence, given that $|Y_i| \geq 2$ and $\lambda = 6/23$, we have
		\begin{equation*}
			\sum_{r \in R_i}z^*_r
			\overset{\eqref{eq:dual_z}}{\leq} \sum_{r\in R_i} v_r 
			< 2\lambda + \lambda |Y_i|
			\leq 2\lambda|Y_i| 
			\leq (1 - 4\lambda/3)|Y_i|
			= \sum_{p\in P_i} y^*_p.
		\end{equation*}

		\item 
		$x_i$ is a thin edge and $|Y_i| = 1$.  $P_i$ contains only one player, so $\sum_{p\in P_i}y^*_p = 1 - 4\lambda/3$ by~\eqref{eq:dual_y}.  Let $e_i$ denote the only blocking edges in $Y_i$.  Note that $e_i$ must be a thin edge.  It suffices to show 
		\begin{equation*}
			\sum_{\substack{\text{$r$ covered} \\\text{by $x_j$ or $e_j$}}}z^*_r \leq 1 - 4\lambda/3.
		\end{equation*} 
\end{enumerate}

\begin{enumerate}[{Case 3.}1]
		\item $x_i$ covers at least one thin resource of value at most $\lambda/2$. 
		Due to the minimality of $x_i$, we have
		\begin{equation*}
			\sum_{\text{$r$ covered by $x_i$}}v_r < \lambda + \lambda/2 = 3\lambda/2.
		\end{equation*}
		Due to the minimality of $e_i$, $e_i$ has less than $\lambda$ worth of resources that are not covered by $x_i$, that is,
		\begin{equation*}
			\sum_{\substack{\text{$r$ covered by} \\ \text{$e_i$ but not $x_i$} }}v_r < \lambda.
		\end{equation*}
		Putting things together, we get
		\begin{align*}
	 	\sum_{\substack{\text{$r$ covered} \\\text{by $x_j$ or $e_j$}}}z^*_r 
	 	\overset{\eqref{eq:dual_z}}{\leq} \sum_{\substack{\text{$r$ covered} \\\text{by $x_j$ or $e_j$}}}v_r 
	 	&= \sum_{\text{$r$ covered by $x_i$}} v_r + \sum_{\substack{\text{$r$ covered by} \\ \text{$e_i$ but not $x_i$} }} v_r \\
	 	&< 3\lambda/2 + \lambda\\
	 	& = 5\lambda/2\\
	 	& = 1 - 4\lambda/3.
		\end{align*}

		\item $e_i$ covers at least one thin resources of value at most $\lambda/2$.  This case can be handled by the same analysis in case 3.1 with the roles of $x_i$ and $e_i$ switched.

		\item All the resources covered by $x_i$ or $e_i$ have value greater than $\lambda/2$.  Two such resources already have total value greater than $\lambda$. Therefore, by the minimality of $x_i$ and $e_i$, each of them covers at most two resources.  Since $e_i$ blocks $x_i$, they have at least one resource in common.  As a result, $x_i$ and $e_i$  covers at most three resources together.  Recall that $x_i$ and $e_i$ are thin edges, 
		so every resource $r$ covered by them is thin and $z^*_r \leq 5\lambda/6$ according to~\eqref{eq:dual_z}. Therefore,  
		\begin{equation*}
			\sum_{\substack{\text{$r$ covered} \\\text{by $x_j$ or $e_j$}}}z^*_r 
			\leq 5\lambda/6 * 3  
			= 5\lambda/2 
			= 1 - 4\lambda/3.
		\end{equation*}
\end{enumerate}
\end{proof}

\begin{lemma}
The algorithm terminates in finite steps.
\label{lem: terminate}
\end{lemma}
\begin{proof}
With respect to the sequence of blockers $(B_1, B_2, \ldots, B_\ell)$, 
we define a signature vector $(|Y_1|, |Y_2|, \ldots, |Y_{\ell}|, \infty)$.  The signature vector evolves as the sequence is updated by the algorithm.  After each invocation of {\sc Build}, the signature vector decreases lexicographically because a new second to last entry is gained.  After each invocation of {\sc Contract}, the signature decreases lexicographically because it becomes shorter, and the second to last entry decreases by at least $1$.  By Lemma~\ref{lem:invar}, $Y_1, Y_2, \ldots, Y_\ell$ are subsets of $M$ that are mutually disjoint. Hence, $|Y_1| + \cdots + |Y_\ell| \leq |M| \leq m$.  Therefore, the number of distinct signature vectors is at most the number of distinct partitions of the integer $m$, which is bounded from above by $O(e^{\sqrt{m}})$~\cite{HR1918}. As a result, the algorithm terminates after at most $O(e^{\sqrt{m}})$ invocations of {\sc Build} and {\sc Contract}.
\end{proof}

By Lemma~\ref{lem:never-stuck} and Lemma~\ref{lem: terminate}, we can always increase the size of the current matching $M$ until it becomes a perfect matching. This proves Theorem~\ref{thm:matching}.  Given a perfect matching of $G$, every player already gets at least $\lambda$ worth of resources.  For the resources not used by the perfect matching, we can allocate them arbitrarily since they have non-negative values.  This results in an allocation where each player receives at least $\lambda$ worth of resources, and Theorem~\ref{thm:main} is proved.

\section{Conclusion}
We have showed that the integrality gap of configuration LP for restricted max-min fair allocation is at most $23/6$, improving the previous upper bound of $4$. In the analysis of the construction algorithm, one of the bottlenecks is the irregularity of the value of the thin edges and the configurations, which results from the irregularity of the value of thin resources.  We partially resolve this issue by using the dual of the configuration LP.  We believe that the upper bound can be further improved if one can derive a better analysis.  Meanwhile, the construction algorithm itself has a few places that can be improved: (1) when there are more than one addable edges, it arbitrarily picks one. Does it help if we have some preference over addable edges?  (2) it restricts the resources of an addable edge to be completely disjoint from those that are already covered by the blockers in the sequence. This is unnecessary in some situations.  A more sophisticated construction may lead to a better bound.

\bibliography{paper}

\end{document}